\documentclass[10pt,twocolumn,twoside]{IEEEtran}
\usepackage{cite}
\usepackage{algorithmic}
\usepackage{textcomp}
\usepackage{epsfig}
\usepackage[utf8]{inputenc}
\usepackage{amsmath}
\usepackage{amssymb}
\usepackage{amsfonts}
\usepackage{amsthm}
\usepackage{siunitx}
\usepackage{psfrag}

\usepackage{enumitem}
\usepackage{xcolor}
\usepackage{notoccite}
\usepackage{threeparttable}
\usepackage{verbatim}
\usepackage{epstopdf}
\usepackage{todonotes}
\usepackage{graphicx, subfigure}
\usepackage{stfloats}
\usepackage{hyperref}
\usepackage{psfrag}

\newtheorem{Remark}{Remark}
\newtheorem{Definition}{Definition}
\newtheorem{Lemma}{Lemma}
\newtheorem{Corollary}{Corollary}

\newtheorem{Proposition}{Proposition}

\newcommand{\sN}{\mathcal{N}}
\newcommand{\sL}{\mathcal{L}}

\newcommand{\abs}[1]{\left\vert #1 \right\vert}
\newcommand{\norm}[1]{\left\Vert #1 \right\Vert}
\newcommand{\R}{{\mathbb R}}  

\newcommand{\black}[1]{{\color{black} #1}}

\begin{document}
\title{Scalability in nonlinear network systems affected by delays and disturbances}
\author{Shihao Xie, \IEEEmembership{Student Member, IEEE}, Giovanni Russo, \IEEEmembership{Senior Member, IEEE}, Richard H. Middleton, \IEEEmembership{Fellow, IEEE}\thanks{\textbf{This is an authors’ version of the work that is published in IEEE Transactions on Control of Network Systems, Volume: 8, Issue: 3, September 2021. Changes were made to this version by the publisher prior to publication. The final version of record is available at \url{https://doi.org/10.1109/TCNS.2021.3058934}}}\thanks{S. Xie (e-mail: \url{shihao.xie1@ucdconnect.ie}) is with the School of Electrical and Electronic Engineering, University College Dublin, Belfield, Dublin $4$, Ireland. G. Russo (e-mail: \url{giovarusso@unisa.it}) is with the Department of Information and Electric Engineering and Applied Mathematics, University of Salerno, Fisciano (SA) $84084$, Italy. R. H. Middleton (e-mail: \url{richard.middleton@newcastle.edu.au }) is with the School of Electrical Engineering and Computing, University of Newcastle, Callaghan NSW $2308$, Australia}}
\maketitle
\begin{abstract}
This paper is concerned with the study of scalability in nonlinear heterogeneous networks affected by communication delays and disturbances.  After formalizing the notion of scalability, we give two sufficient conditions to assess this property. Our results can be used to study  {leader-follower} and leaderless networks and also allow to consider the case when the desired configuration of the system changes over time. We show how our conditions can be turned into design guidelines to guarantee scalability and illustrate their effectiveness via numerical examples.
\end{abstract}
\thispagestyle{empty}
\pagestyle{plain}
\section{Introduction}

Network systems have considerably evolved, increasing not only their scale but also the complexity of their topology \cite{8684336}. Examples of  large-scale networks include autonomous vehicles \cite{STUDLI2017157},  robotic formations \cite{4200874}, neural networks \cite{7336524}. There is  no surprise that a large body of literature devoted to the study of collective behaviours has emerged, with e.g. consensus and synchronization attracting much research attention \cite{1333204,7206553,8528873}. 

In this context, a key challenge is the design of protocols that do not just guarantee \textit{stability} (fulfillment of desired behaviour for a fixed network)  but also that the network is {\em scalable}. \textcolor{black}{We use \textit{scalability} to denote the preservation of desired stability properties (to be defined more formally in Section \ref{sec:formalize_problem}) uniformly with respect to the size of the network}. {Scalability} is then a fundamental requirement for network systems  spanning from e.g. platoons of vehicles to neural networks. \textcolor{black}{For example, as noted in e.g. \cite{folli}, there is an intrinsic limit (this limit  appears to be approximately 14\% of the network size) for recurrent networks that precludes them to { store} an arbitrarily large number of memory patterns. Hence,  to increase memory capacity in general requires  increasing the number of neurons. However, if not properly designed, adding new layers/{neurons} might lead to the amplification of disturbances/biases as these propagate through the network \cite{pmlr-v120-bonassi20a}. Designing the network so that this is not just stable but also scalable avoids the onset of this undesired behaviour.}  Motivated by this, we: (i) introduce a notion of scalability for networks affected by delays; (ii) give sufficient conditions to assess this property; (iii) show how our approach can be turned into design guidelines for scalability. 

\subsubsection*{Related work} 
the study of how disturbances propagate within a network is a central topic for autonomous vehicle-following systems. In the context of platoons much research effort has been devoted to the study of string stability \cite{swaroop1996string}. The key idea behind the several definitions of string stability proposed in the literature \cite{FENG201981} is that of giving upper bounds on certain deviations of the individual agents from a reference that are independent on the number of vehicles. See e.g. \cite{swaroop1996string,STUDLI2017157,FENG201981,knorn2014passivity,MONTEIL2019198} for recent results and a survey of the related literature. In the above works, results are obtained under the assumption that the network system is delay-free and some extensions to strings affected by communication delays are given in e.g. \cite{6891349,7839205} for homogeneous, disturbance-free, linear systems.  Works on scalability for networks with arbitrary topologies are sparse when compared to works on string stability. For networks with linear agents, results include \cite{6209387,8684336} where network coherence is characterized as a function of the number of its agents,\cite{7442805} where performance deterioration in networks subject to external stochastic disturbances are considered, \cite{8070997} where certain network performance metrics are studied as a function of the number of edges. Other related works include these on leader-to-formation stability, see e.g. \cite{1303690}, that was introduced to characterize the behaviour of (disturbance-free) formations with respect to the inputs provided by leaders and these on mesh stability, see e.g. \cite{Pant2001}, that offer a generalization of string stability to (linear and distubance-free) networks with regular topologies. Other results include \cite{articleBesselink}, where sufficient conditions for the scalability of delay-free leaderless networks with homogeneous agents interacting over regular topologies are introduced. \textcolor{black}{Our proofs leverage contraction-theory arguments for time-delayed systems. We recall  \cite{1618853} which shows, using the Euclidean metric, how contraction is preserved through certain time-delayed communications, and  \cite{9029867} which, by extending the integral quadratic constraints method,  provides a tighter characterisation of delays. Finally, we make use of max-separable metrics and we refer the reader to e.g. \cite{8561231}, where conditions for the synthesis of distributed controls are given by using separable metric structures.}

\subsubsection*{Statement of the contributions} \textcolor{black}{our results are based on the key observation that, if a system is contracting with respect to a max-separable metric, then this is also scalable.  We then give a set of sufficient conditions (independent on the bounds of the delays) for scalability of network systems}  consisting of  heterogeneous nonlinear agents that communicate via possibly nonlinear protocols. The agents are affected by external disturbances and communication delays. In particular:
\begin{itemize}
\item for these delayed networks, we formalize the notions of $\sL_{\infty}$-scalable-Input-to-State {Stable} and $\sL_{\infty}$-scalable-Input-Output {Stable} ($\sL_{\infty}$-sISS and $\sL_{\infty}$-sIOS) networks;
\item we give two sufficient conditions to assess these properties. To the best of our knowledge, these are the first results that tackle the problem of guaranteeing scalability of nonlinear networks of heterogeneous agents affected by both disturbances and delays. Moreover, our results can be used to study both {leader-follower} and leaderless networks and also allow to consider the case where the desired configuration of the system changes over time;
\item we show that our conditions can be turned into design guidelines for scalability. We do so by first showing how our approach can be used to design   protocols that, while guaranteeing the tracking of a time-varying speed profile,  ensure scalability for a network of mobile agents. We then show how the approach can be used to devise conditions on the activation functions (and their weights) to guarantee scalability of certain  neural networks. \textcolor{black}{Motivated by applications such as associative memory, where it is of interest to study stability of equilibria when the network is forced by a constant input \cite{ZHOU201644}, the case we consider is when the recurrent network receives as input a constant vector, {possibly affected by a bias,} and the desired output is an equilibrium point (see also \cite{8329007,FAYDASICOK2020288}).} To the best of our knowledge, these are  the first results that explicitly address scalability in neural networks.
\end{itemize}

\section{Mathematical preliminaries}\label{sec:math}
Let $A$ be a $m \times m$ real matrix. We denote by $\norm{A}_p$ the matrix norm induced by the $p$-vector norm $\abs{\cdot}_p$. We recall that (see e.g. \cite{Vid_93}) the matrix measure induced by $\abs{\cdot}_p$ is defined as $\mu_p(A):= \lim\limits_{h \rightarrow 0^+}\frac{1}{h}(\norm{I+h A}_p-1)$. In this work, we make use of $\mu_2(A):=\lambda_{\max}(A^T+A)/2$ and {$\mu_\infty(A):=\max_i[a_{ii}+\sum_{j \ne i}\abs{a_{ij}}]$}. We also denote by $\sigma_{\min}(A)$ ($\sigma_{\max}(A)$) the smallest (largest) singular value of $A$. Given the piece-wise continuous signal $d_i(t)$: $[0,+\infty) \rightarrow \mathbb{R}^m$, we let $\Vert d_i(\cdot) \Vert_{\mathcal{L}_{\infty}}:= \sup_t\abs{d_i(t)}_2$. {Let {$f(x_1,\ldots,x_N)$} be a smooth function in all its arguments, with $x_i\in\R^n$. Then, {$\partial_if(x_1,\ldots,x_N)$}  denotes the partial derivative {$\frac{\partial f(x_1,\ldots,x_N)}{\partial x_i}$}. Throughout the paper, the $m$-dimensional identity matrix is denoted by $I_m$ while the $m\times n$ zero matrix is $0_{m\times n}$. We recall that a continuous function ${\kappa}: \R_+ \rightarrow \R_+$ is said to be a class-$\mathcal{K}$ function  if it is strictly increasing and ${\kappa(0)}=0$. It is said to belong to class-$\mathcal{K}_\infty$ if  ${\kappa(r)} \rightarrow +\infty$ as $r \rightarrow +\infty$. A continuous function $\beta: \R_+ \times \R_+ \rightarrow \R_+$ is said to be a class-$\mathcal{KL}$ function if, for any fixed $s$, $\beta(\cdot,s)$ is of class-$\mathcal{K}$ and, for each fixed $r$, $\beta(r,\cdot)$ is decreasing and $\beta(r,s) \rightarrow 0$ as $s \rightarrow \infty$. 
\subsubsection*{Useful results}
\textcolor{black}{The next result follows directly from \cite{5717887}.} We let $\abs{\cdot}_S$ and $\mu_S(\cdot)$ be, respectively, any $p$-vector norm and its induced matrix measure on $\mathbb{R}^N$. In particular, the norm $\abs{\cdot}_S$ is monotone, i.e. for any vector $x,y \in \R_{+}^N$, $x\le y$ implies that $\abs{x}_S \le \abs{y}_S$ where the inequality $x\le y$ is component-wise.
\begin{Lemma}\label{lem:matrix norm}
Consider the vector $\eta:=[\eta_1^T,\ldots,\eta_N^T]^T$, $\eta_i \in \mathbb{R}^n$. We let $\abs{\eta}_G := \abs{\left[ \abs{\eta_1}_{G_1},\ldots,\abs{\eta_N}_{G_N}\right]}_S$, with $\abs{\cdot}_{G_i}$ being norms on $\R^n$, and denote by $\norm{\cdot}_G, \mu_G(\cdot)$ ($\norm{\cdot}_{G_i}, \mu_{G_i}(\cdot)$) the matrix norm and measure induced by $\abs{\cdot}_G$ ($\abs{\cdot}_{G_i}$). Finally, let: 
\begin{enumerate}
\item $A:=({A_{ij}})_{i,j=1}^N \in \mathbb{R}^{nN \times nN}$, $A_{ij} \in \mathbb{R}^{n \times n}$; 
\item $\hat{A}:=(\hat{A}_{ij})_{i,j=1}^N \in \mathbb{R}^{N \times {N}}$, with $\hat{A}_{ii}:=\mu_{G_i}(A_{ii})$ and $\hat{A}_{ij}:=\norm{A_{ij}}_{G_{i,j}}$, {$\norm{A_{ij}}_{G_{i,j}}:=\sup_{\abs{x}_{G_i}=1}\abs{A_{ij}x}_{G_j}$};
\item $\bar{A}:=(\bar{A}_{ij})_{i,j=1}^N \in \mathbb{R}^{N \times N}$, with $\bar{A}_{ij}:=\norm{A_{ij}}_{G_{i,j}}$. 
\end{enumerate}
Then: (i) $\mu_G(A) \le \mu_S(\hat{A})$; (ii) $\norm{A}_G \le \norm{{\bar A}}_S$. 
\end{Lemma}
The next proposition follows from Theorem $2.4$ in \cite{WEN2008169}.
\begin{Lemma}\label{prop:halanay}
Let $u:[-\tau_0,+\infty)\rightarrow\R_+$ , $\tau_0<+\infty$ and assume that
$$D^+u(t) \le au(t)+b \sup_{t-\tau(t) \le s \le t}u(s)+c, \ \ \ t\ge 0$$
with: (i) $\tau(t)$ being bounded and non-negative, i.e. $0{\le}\tau(t)\le\tau_0$, $\forall t$; (ii) $u(t)=\abs{\varphi(t)}$, $\forall t\in [-\tau_0,0]$ where $\varphi(t)$ is bounded in $[-\tau_0,0]$; (iii) $a < 0$, $b \ge 0$ and $c \ge 0$. Assume that there exists some $\sigma>0$ such that $a+b \le -\sigma <0, \forall t\ge 0$. Then:
$$u(t) \le \sup_{-\tau_0 \le s \le 0}u(s)e^{-\hat \lambda t}+\frac{c}{\sigma}$$
where $\hat \lambda:=\inf_{t\ge 0}\{\lambda|\lambda(t) + a+be^{\lambda(t)\tau(t)}=0\}$ is positive.
\end{Lemma}

\section{Statement of the control problem}\label{sec:formalize_problem}
We consider a network of $N$ heterogeneous agents, possibly receiving inputs from a set of $L$ leaders. The dynamics of the $i$-th agent, $i=1,..., N$, is modelled via
\begin{align}
    \dot{x}_i&={f_i(x_i,t)}+u_i(t)+b_i(x_i,t)d_i(t), \ \ \ t\ge 0 \label{dynamics_1}\\    
    y_i&=g_i(x_i)\label{dynamics_2}
\end{align}
with $x_i \in \mathbb{R}^n$, $u_i(t) \in \mathbb{R}^n$ being the control protocol, $d_i(t)$ being an $n$-dimensional signal modelling a (deterministic) external disturbance on the agent. In the  dynamics, $b_i:\R^n\times\R_+\rightarrow\R^{n\times n}$ models a time- and state-dependent disturbance intensity function. The function $f_i:\mathbb{R}^n {\times \R_+} \rightarrow \mathbb{R}^n$ is the intrinsic dynamics of the agent, $g_i:\R^n \rightarrow \R^m$ is the output function. The noise intensity is bounded, i.e. $\max_{x_i,t}{\norm{b_i(x_i,t)}_2}\le \bar b$, $\forall i$ and {$f_i(\cdot,\cdot)$}, $b_i(\cdot,\cdot)$, $g_i(\cdot)$ are smooth in  their arguments. 
We consider protocols of the form
\begin{equation}\label{control_protocol}
\begin{split}
u_i(t)&=\sum_{j \in \sN_i} {h_{ij}^{(\tau)}} \big(x_i(t-\tau(t)),x_j(t-\tau(t)),t\big)\\
&+\sum_{j \in \sN_i} {h_{ij}}\big(x_i(t),x_j(t),t\big)\\
&+\sum_{l\in \sL_i} {\underline{h}_{il}^{(\tau)}}\big(x_i(t-\tau(t)),x_l(t-\tau(t)),t\big)\\
&+\sum_{l\in \sL_i}{\underline{h}_{il}}\big(x_i(t),x_l(t),t\big)
\end{split}
\end{equation}
with $\tau(t)\le \tau_0$ $\forall t$,} $x_i(s)=\varphi_i(s)$, { $\varphi_i(s):\R\rightarrow\R^n$ being continuous and} bounded $\forall s \in [-{\tau_0},0]$, $\forall i=1,\ldots,N$. In (\ref{control_protocol}): (i) $\mathcal{N}_i$ denotes the set of neighbours of agent $i$ and $\sL_i$ is the set of leaders to which the $i$-th agent is possibly connected; (ii) the functions $h_{ij}:\R^n\times \R^n\times\R_+  \rightarrow\R^n$ and $\underline{h}_{il}:\R^n\times\R^{n}\times\R_+  \rightarrow\R^n$ are the {\em delay free} inter-agent and agent-leader coupling functions respectively. Analogously, {$h_{ij}^{(\tau)}:\R^n\times \R^n\times\R_+  \rightarrow\R^n$} and {$\underline{h}_{il}^{(\tau)}:\R^n\times \R^{n}\times\R_+ \rightarrow\R^n$} are coupling functions for the delayed information. That is, in (\ref{control_protocol}) we have separate couplings for the delay-free and delayed communication. 
\begin{Remark}
Situations where there is an overlap between a delay-free and delayed communication naturally arise in a number of applications. For example, in platooning \cite{PETERS201464}, certain states (e.g. separation obtained by radar) from the neighbours might be available at an agent with a negligible/no delay. Other information, such as separation from more distant agents, or neighbour control actions, may require communications and be subject to measurement/processing delays. Protocols of the form of (\ref{control_protocol}) can also be used to study the popular neural network model of e.g \cite{4666771}. {In this case,  { agents} are neurons and a delay-free communication coupling models an activation function from the closest neurons.  The delayed communication can instead model interactions from neurons farther away in the network.}
\end{Remark}
\subsection{Control goal}
As in \cite{MONTEIL2019198} we state our control goal in terms of a {\em desired} solution of the unperturbed dynamics of (\ref{dynamics_1}) - (\ref{control_protocol}). We let $x^d(t)=[x_1^d(t)^T,\ldots,x_N^d(t)^T]^T$, $\dot{x}^d_i(t)=f_i(x_i^d,t)$,  be the desired state/solution of the network when there are no disturbances and $y^d(t)=[y_1^d(t)^T,\ldots,y_N^d(t)^T]^T$, with $y^d_i(t)=g_i(x_i^d(t))$, be the desired output. Our goal is to design  (\ref{control_protocol}) so that (\ref{dynamics_1}) - (\ref{dynamics_2})
is either {$\sL_{\infty}$-scalable-Input-to-State Stable ($\sL_{\infty}$-sISS) or $\sL_{\infty}$-scalable-Input-Output Stable ($\sL_{\infty}$-sIOS}):
\begin{Definition}\label{def:ISFS}
The closed loop network system (\ref{dynamics_1}) - (\ref{control_protocol}) is:
\begin{itemize}
\item $\mathcal{L}_{\infty}$-scalable-Input-to-State  Stable ($\mathcal{L}_{\infty}$-sISS) if there exists some class-$\mathcal{KL}$ function, $\beta$, and class-$\mathcal{K}$ function, $\gamma$,  such that, for all $t\ge 0$, ${\max_i} \vert x_i(t)-x_i^d(t) \vert_2 \le  \beta ({\max_i}\sup_{-\tau_0\le s \le 0}\vert x_i(s)-x_i^d(s) \vert_2,t ) + \gamma ({\max_i}\Vert d_i(\cdot) \Vert_{\mathcal{L}_{\infty}})$, $\forall N$;
\item $\mathcal{L}_{\infty}$-scalable-Input-Output Stable ($\mathcal{L}_{\infty}$-sIOS) if there exists some class-$\mathcal{KL}$ function, {$\kappa$}, and class-$\mathcal{K}$ function, $\gamma$,  such that, for all $t\ge 0$, ${\max_i} \vert y_i(t)-y_i^d(t) \vert_2 \le  {\kappa} ({\max_i}\sup_{-\tau_0\le s \le 0}\vert x_i(s)-x_i^d(s) \vert_2,t ) + \gamma ({\max_i}\Vert d_i(\cdot) \Vert_{\mathcal{L}_{\infty}})$, $\forall N$,
\end{itemize}
where $x_i(t)$ is a solution of the system with $x_i(s)=\varphi_i(s),s \in [-{\tau_0},0],i=1,\ldots,N$ and $x_i^d(s)=x_i^d(0)$, $s \in [-\tau_0,0]$.
\end{Definition}
\begin{Remark}
We say that the {network system} is $\mathcal{L}_{\infty}$-sISS ($\mathcal{L}_{\infty}$-sIOS) if the above definition is fulfilled. The upper bounds in the definition give an estimate on the maximum {\em deviation} of state (output) from the desired configuration/solution. The functions {$\kappa$}, $\beta$ and $\gamma$ are not dependent on the number of agents in the network and \textcolor{black}{such invariance of the bounds w.r.t. $N$ is a key feature that differentiates scalability from the classic notions of ISS and IOS.} This guarantees that disturbances will not grow without bound as new agents are added, \textcolor{black}{thus supporting the possibility of adding new agents.}
\end{Remark}
\section{Technical results}\label{sec:results}
Our first result is a sufficient condition guaranteeing $\mathcal{L}_{\infty}$-sISS of (\ref{dynamics_1}) - (\ref{control_protocol}). Whenever it is clear from the context, we omit the dependency of the state variables on the time.
\begin{Proposition}\label{prop: ISFS}
Consider network  (\ref{dynamics_1}) - (\ref{control_protocol}) with $y_i(t) = x_i(t)$. Assume that, $\forall i=1,\ldots,N$ and $\forall t\ge 0$, the following conditions are satisfied for some $0<\underline{\sigma}<\bar{\sigma}\textcolor{black}{<+\infty}$:
\item[\bf{(i)}] 
\begin{equation*}
\begin{split}
& h_{ij}(x_i^d(t),x_j^d(t),t)=h_{ij}^{(\tau)}(x_i^d(t-\tau(t)),x_j^d(t-\tau(t)),t)\\
&=\underline{h}_{il}(x_i^d(t),x_l(t),t)=\underline{h}_{il}^{(\tau)}(x_i^d(t-\tau(t)),x_l(t-\tau(t)),t)=0
\end{split}
\end{equation*}
\item[\bf{(ii)}] 
\begin{equation*}
\begin{split}
&\mu_2\big( \partial_1{f_i(x_i,t)}+\sum_{l \in \sL_i} \partial_1 {\underline{h}_{il}}(x_i,x_l,t)+\sum_{j \in \sN_i} \partial_1 {h_{ij}}(x_i,x_j,t) \big)\\
&+\sum_{j\in\mathcal{N}_i}\norm{\partial_2 {h_{ij}}(x_i,x_j,t)}_2 \le -\bar \sigma,\ \ \forall x_i,x_j,x_l\in\R^n
\end{split}
\end{equation*}
\item[\bf{(iii)}] 
\begin{equation*}
\begin{split}
&\norm{\sum_{l \in \sL_i}  \partial_1{\underline{h}_{il}^{(\tau)}}(x_i,x_l,t) + \sum_{j \in \sN_i} \partial_1{h_{ij}^{(\tau)}}(x_i,x_j,t)}_2\\
&+\sum_{j \in \sN_i}\norm{\partial_2 {h_{ij}^{(\tau)}}(x_i,x_j,t)}_2 \le \underline \sigma, \ \ \forall x_i,x_j,x_l\in\R^n
\end{split}
\end{equation*}
Then, the system is $\sL_{\infty}$-sISS. In particular, we have $\forall t\ge0$:
\begin{equation}\label{eqn:upper_bound}
\begin{split}
\max_i\abs{x_i(t)-x^d_i(t)}_2 &\le \max_i\sup_{-\tau_0 \le s \le 0}\abs{x_i(s)-x_i^d(s)}_2e^{-\hat \lambda t}\\ 
&+\frac{\bar b}{\bar \sigma-\underline \sigma}\max_i\norm{d_i(\cdot)}_{\sL_\infty}, \forall N
\end{split}
\end{equation}
where ${0<}\hat \lambda=\inf_{t\ge 0}\{\lambda|\lambda(t) - \bar \sigma+\underline \sigma e^{\lambda(t)\tau(t)}=0\}$, $x_i(t)$ is a solution of the system with initial value $x_i(s)=\varphi_i(s)$, $s\in [-\tau_0,0]$, $i=1,\ldots,N$ {and  $x_i^d(s)=x_i^d(0)$, $s \in [-\tau_0,0]$}.
\end{Proposition}
In what follows, when we state the other scalability results, we omit that $x_i(s)=\varphi_i(s)$, $\forall s\in [-\tau_0,0]$ {and that $x_i^d(s)=x_i^d(0)$, $\forall s \in [-\tau_0,0]$ as this is clear from the context. Before presenting the proof of the above result we note the following.

\begin{Remark}
{\textcolor{black}{Intuitively, condition (i) implies that, at the desired solution,  $u_i(t)=0$. This rather common (see e.g. \cite{1333204,4200874,STUDLI2017157}) condition} guarantees that the desired solution is a solution of the unpertubed dynamics.} Conditions {(ii)} and (iii) give upper bounds on the matrix measure and matrix norm of the Jacobian of the controlled network system. These conditions imply that the Jacobian of the intrinsic dynamics and of the delay-free part of the protocol have a matrix measure that is {\em sufficiently negative} to balance the presence of the delays. \textcolor{black}{As we shall see, the conditions imply the existence of a norm in which the delay-free part of the dynamics has a negative matrix measure that is small enough to compensate the delays.}
\end{Remark}

\begin{Remark}\label{rem:delays}
Interestingly, if the conditions of the proposition are satisfied, the scalability of the formation is guaranteed for any bounded delay. While scalability is guaranteed independently on the delay, the {\em convergence rate} $\hat\lambda$ depends on $\tau(t)$. Also, we do not require any assumption on the differentiability of the delays. In this sense, our results relax an assumption made to study stability in e.g. \cite{7486039} and related references. 
\end{Remark}
\black{
\begin{Remark}
As in \cite{articleBesselink}, the definition of scalability used in this paper is independent on the topology of the interconnections between the agents. Also, in accordance to \cite{articleBesselink}, designing a given network so that each agent has the number of neighbours independent on the network size can be leveraged to satisfy conditions (ii) and (iii) of Proposition \ref{prop: ISFS}. Unfortunately, in general adding connections between agents can lead to a loss of the scalability property (see also Figure \ref{nonscalable} \textcolor{black}{and the related discussion in Section \ref{sec:robotics_example}}).
\end{Remark}}
We are now ready to give the proof for Proposition \ref{prop: ISFS}.
\begin{proof} 
We start with noting that, following condition (i), $x_i^d(t)$ satisfies {$\dot{x}_i^d(t)=f_i(x_i^d(t),t)$}. Hence:
\begin{align*}
&\dot{x}_i(t)-\dot{x}_i^d(t) \nonumber \\ 
=&{f_i(x_i(t),t)-f_i(x_i^d(t),t)} \nonumber\\
+&\sum_{l \in \sL_i}\black{\underline{h}_{il}}\big(x_i(t),x_l(t),t\big)-\sum_{l \in \sL_i}\black{\underline{h}_{il}}\big(x_i^d(t),x_l(t),t\big) \nonumber\\
+&\sum_{l \in \sL_i} \black{\underline{h}_{il}^{(\tau)}}\big(x_i(t-\tau(t)),x_l(t-\tau(t)),t\big) \nonumber\\ 
-&\sum_{l \in \sL_i} \black{\underline{h}_{il}^{(\tau)}}\big(x_i^d(t-\tau(t)),x_l(t-\tau(t)),t\big) \nonumber\\
+&\sum_{j \in \sN_i}\black{h_{ij}}\big(x_i(t),x_j(t),t\big)-\sum_{j\in \sN_i}\black{h_{ij}}\big(x_i^d(t),x_j^d(t),t\big) \nonumber\\
+&\sum_{j \in \sN_i} \black{h_{ij}^{(\tau)}}\big(x_i(t-\tau(t)),x_j(t-\tau(t)),t\big) \nonumber\\
-&\sum_{j\in \sN_i} \black{h_{ij}^{(\tau)}}\big(x_i^d(t-\tau(t)),x_j^d(t-\tau(t)),t\big)+b_i(x_i,t)d_i(t)
\end{align*}
Let $z(t)=[z_1^T(t),\ldots,z_N^T(t)]^T$, $z_i(t)=x_i(t)-x_i^d(t)$. Then, the dynamics of $z(t)$ can be written as (see e.g.  \cite{Desoer}):
\begin{align}\label{eqn:error_dyn}
\dot{z}(t)=A(t)z(t)+H(t)z(t-\tau(t))+B(x,t)d(t)
\end{align}
where $A(t)$ is a $nN\times nN$ matrix consisting of the $n\times n$ blocks defined, $\forall i,j=1,\ldots,N$, as follows: (i) $
A_{ii}(t):=\int_{0}^1 \partial_1{f_i(e_i(\eta),t)}d\eta+\sum_{l \in \sL_i}\int_{0}^1 \partial_1 \black{\underline{h}_{il}}(e_i(\eta),x_l,t)d\eta + \sum_{j \in \sN_i}\int_{0}^1 \partial_1 \black{h_{ij}}(e_i(\eta),e_j(\eta),t)d\eta
$; (ii) $A_{ij}(t):=\int_{0}^1 \partial_2 \black{h_{ij}}(e_i(\eta),e_j(\eta),t)d\eta$, where $e_i(\eta) := \eta x_i+(1-\eta)x_i^d$ and where the dependency of the state variables on time has been omitted. Also, in (\ref{eqn:error_dyn}), $B(x,t)$ is the $nN\times nN$ block-diagonal matrix having $b_i(x_i,t)$ on the main diagonal, $d(t)=[d_1^T(t),\ldots,d_N^T(t)]^T$ and $H(t)=(H_{ij})_{i,j=1}^N$ is the block matrix consisting of the $n\times n$ blocks:(i) $H_{ii}(t):=\sum_{l \in \sL_i} \int_{0}^1  \partial_1\black{\underline{h}_{il}^{(\tau)}}(e_i(\eta),x_l,t)d\eta + \sum_{j \in \sN_i} \int_{0}^1  \partial_1\black{h_{ij}^{(\tau)}}(e_i(\eta),e_j(\eta),t)d\eta$; (ii) $H_{ij}(t):=\int_{0}^1 \partial_2 \black{h_{ij}^{(\tau)}}(e_i(\eta),e_j(\eta),t)d\eta$, where again we omitted the explicit dependency of the state variables on time and $e_i(\eta) := \eta x_i+(1-\eta)x_i^d$. We now study the error dynamics in (\ref{eqn:error_dyn}). To this aim we make use of {Lemma \ref{prop:halanay}} and define $\abs{z(t)}_G:=\vert [\abs{z_1(t)}_2,\ldots,\abs{z_N(t)}_2] \vert_\infty$, which can be easily seen to be a vector norm. By taking the Dini derivative of $\abs{z(t)}_G$, from (\ref{eqn:error_dyn}) we get:
\begin{equation*}
\begin{split}
&D^+\abs{z(t)}_G := \limsup_{h\rightarrow 0^+}\frac{1}{h}(\abs{z(t+h)}_G-\abs{z(t)}_G)\\
= &\limsup_{h\rightarrow 0^+}\frac{1}{h}\bigg(\vert z(t)+hA(t)z(t)+hH(t)z(t-\tau(t))\\
&+hB(x,t)d(t)\vert_G-\abs{z(t)}_G \bigg)\\
\le & \mu_G\big(A(t)\big)\abs{z(t)}_G+\norm{H(t)}_G \sup_{t-\tau(t) \le s \le t}\abs{z(s)}_G\\
&+ \bar b \max_i\norm{d_i(\cdot)}_{\sL_\infty}
\end{split}
\end{equation*} 
which was obtained by means of the triangle inequality and by using the fact that, from the definition of $\abs{\cdot}_G$ and the boundedness of $b_i(x_i,t)$, $\norm{B(t,x)}_G\le \max_i\sup_{x_i,t}{\norm{b_i(x_i,t)}_2} \le \bar b$. In order to apply {Lemma \ref{prop:halanay}}, we need to find the upper bound of $\mu_G(A(t))$ and $\norm{H(t)}_G$. This can be computed via {Lemma \ref{lem:matrix norm}}. Indeed, from such a result it follows that $\mu_G(A(t)) \le \max_i\big\{\mu_2(A_{ii}(t))+\sum_{\textcolor{black}{j\in \sN_i}}\norm{A_{ij}(t)}_2\big\}$ and $\norm{H(t)}_G \le \max_i\big\{\sum_{j=1}^N\norm{H_{ij}(t)}_2\big\}$. Moreover, conditions { (ii)} and { (iii)} imply that $\max_i\big\{\mu_2(A_{ii}(t))+\sum_{\textcolor{black}{j\in \sN_i}}\norm{A_{ij}(t)}_2\big\}\le -\bar{\sigma}$ and $\max_i\big\{\sum_{j=1}^N\norm{H_{ij}(t)}_2\big\}\le \underline{\sigma}$, for some $0<\underline{\sigma}<\bar{\sigma}\textcolor{black}{<+\infty}$. Hence, we get 
\begin{equation*}
\begin{split}
&D^+\abs{z(t)}_G \le \\
&-\bar{\sigma}\abs{z(t)}_G +\underline{\sigma} \sup_{t-\tau(t)\le s\le t}\abs{z(s)}_G + \bar{b}\max_i\norm{d_i(\cdot)}_{\sL_\infty}
\end{split}
\end{equation*}
$\forall N$, which, by means of {Lemma \ref{prop:halanay}}, yields
$$\abs{z(t)}_G \le \sup_{-\tau_0 \le s \le 0}\abs{z(s)}_Ge^{-\hat \lambda t}+\frac{\bar b}{\bar \sigma -\underline \sigma}\max_i\norm{d_i(\cdot)}_{\sL_\infty}$$ $\forall N$, with $0<\hat \lambda=\inf_{t\ge 0}\{\lambda|\lambda(t) - \bar \sigma+\underline \sigma e^{\lambda(t)\tau(t)}=0\}$.
Since $z(t) = x_i(t) - x_i^d(t)$, this completes the proof.
\end{proof}
\textcolor{black}{To further highlight a key difference between Definition \ref{def:ISFS} and the classical notion of ISS, consider the special case of a string of agents arranged in a cascading configuration. In this case, it is well-known  that the cascade of ISS agents is also ISS.  However, this does not guarantee that the network is scalable. Indeed, \cite{7879221} gives a counter-example (see Remark $3$ therein) where perturbations grow without bound for a string of ISS systems when the number of agents grows.} With the next result we give a sufficient condition for $\sL_{\infty}$-sIOS.
\begin{Proposition}\label{prop:IOS}
Consider the closed loop network  (\ref{dynamics_1}) - (\ref{control_protocol}) and assume that the conditions $(i)-(iii)$ of {Proposition \ref{prop: ISFS}} are satisfied and that, additionally, the functions $g_i(\cdot)$ are Lipschitz with Lipschitz constant $k_i$. Then, (\ref{dynamics_1})-(\ref{control_protocol}) is also $\sL_{\infty}$-sIOS.
\end{Proposition}
\begin{proof} The fulfilment of $(i)-(iii)$ of {Proposition \ref{prop: ISFS}} implies the upper bound in (\ref{eqn:upper_bound}). Since the output functions are Lipschitz, we also have that $\abs{y_i-y_i^d}_2:=\abs{g_i(x_i)-g_i(x_i^d)}_2\le k_i\abs{x_i-x_i^d}_2\le k\sup_i\abs{x_i-x_i^d}_2$, where $k:=\max_ik_i$. This, together with (\ref{eqn:upper_bound}) immediately implies the result.
\end{proof}

\section{Using the results to design scalable network systems}\label{sec:examples}
We illustrate how our approach can be effectively used to design scalable networks. The first application we consider is concerned with the design of distributed  protocols for a robotic formation. With the second application we focus on ensuring that certain recurrent neural networks \textcolor{black}{(RNNs)} are scalable.

\subsection{Formation scalability}\label{sec:robotics}
We now consider the problem of designing a control protocol guaranteeing  scalability of the formation for a network of $N$ mobile robots, while tracking a time-varying reference provided by a virtual leader. Each robot is modeled via a non-holonomic unicycle and, in particular, we adapt the popular model from \cite{lawton} by embedding external disturbances:
\begin{align}\label{application}
\begin{split}
    \dot{p}^x_i&=v_i\cos{\theta}_i\\
    \dot{p}^y_i&=v_i\sin{\theta}_i\\
    \dot{v}_i&=\frac{F_i+d^f_i}{m_i}\\
    \dot{\theta}_i&=\omega_i\\
    \dot{\omega}_i&=\frac{Q_i+d^q_i  }{I_i}
\end{split}
\end{align}
In the  model $p_i(t):=[p^x_i(t),p^y_i(t)]^T$ is the inertial position, $v_i(t)$ the linear speed, $\theta_i(t)$ the heading angle, $\omega_i(t)$ the angular velocity, $m_i$ the mass, $F_i(t)$ the applied force input, $Q_i(t)$ the applied torque input, $I_i(t)$ the moment of inertia, $d_i^f(t)$ the external force disturbance and $d_i^q(t)$ the external torque disturbance. In (\ref{application}) the disturbance $d^f_i(t)$ models uncertainties due to e.g. unmodeled friction forces, while $d^q_i(t)$ models external disturbances due to e.g. wind. We aim at controlling the {\em hand} position of the robots and denote by $\eta_i(t)$ the hand position of the $i$-th robot. As shown in the Appendix, the dynamics (\ref{application}) can be feedback linearised, yielding 
\begin{align}
\dot{\chi}_i &= A_i\chi_i + \nu_i(t) + b_i (t)d_i(t) \label{hpmodel_linearized:1}\\
\eta_i &=C_i\chi_i   \label{hpmodel_linearized:2}
\end{align} 
with {control  $\nu_i(t):=\left[\begin{matrix} 0_{2\times1} \\ \bar \nu_i(t)\end{matrix}\right]$}, $\chi_i(t):=[\chi_{i,1},\chi_{i,2},\chi_{i,3},\chi_{i,4}]^T=[p_i^x+l_i\cos\theta_i,p_i^y+l_i\sin\theta_i,v_i\cos\theta_i-l_iw_i\sin\theta_i,v_i\sin\theta_i+l_iw_i\cos\theta_i]^T$, $A_i:=\left[\begin{matrix} 0_{2\times2} & I_2 \\ 0_{2\times2} & 0_{2\times2} \end{matrix}\right]$, $b_i(t):=\left[\begin{matrix} 0_{2\times2} & 0_{2\times2}\\0_{2\times2} & \bar b_{i}(t) \end{matrix}\right]$, $d_i(t):=[0_{1\times2},\bar d_{i}(t)^T ]^T$, $C_i=\left[\begin{matrix} I_2 & 0_{2\times 2} \end{matrix}\right]$ (see the Appendix for the definitions).

In what follows we design the control $\nu_i(t)$ so that the above network system {is $\mathcal{L}_{\infty}$-sIOS and hence it}  has a scalable formation. We consider the popular class of protocols considered in e.g. \cite{weixun} where: (i) communications between robots are affected by delay; (ii) robots have access to a reference trajectory (i.e. hand position and speed) provided by a {\em virtual} leader, $\chi_l(t)$. In particular, $\eta_l(t):=[\chi_{l,1},\chi_{l,2}]^T$ denotes the hand position of the leader at time $t$ and $v_l(t):=[\chi_{l,3},\chi_{l,4}]^T$ is the corresponding smooth speed signal.}
Within the formation, the hand position of the $i$-th robot needs to keep some desired offset from the neighbours and from $\eta_l(t)$ while, at the same time, tracking the acceleration and speed provided by the virtual leader. That is, the desired solution for the $i$-th robot within the formation when there are no perturbations/delays, i.e. $\chi_i^d(t):= [\chi_{i,1}^d,\chi_{i,2}^d,\chi_{i,3}^d,\chi_{i,4}^d]^T$ is such that: (i) the desired hand position of the robot $\eta_i^d(t):=[\chi_{i,1}^d,\chi_{i,2}^d]^T$ keeps the desired offsets; (ii) the corresponding speed is $v_i^d(t)=[\chi_{l,3},\chi_{l,4}]^T$; (iii) satisfies
\begin{equation}\label{eqn:formation_desired_sol}
\dot\chi_i^d = \left[\begin{matrix} 0_{2\times2} & I_2 \\ 0_{2\times2} & 0_{2\times2} \end{matrix}\right]\chi_i^d + \left[\begin{matrix} 0_{2\times 1} \\ \dot v_l(t) \end{matrix}\right]
\end{equation}
\subsubsection*{Protocol design}
we consider control protocols of the form
\begin{align}\label{controlnu}
\bar\nu_i(t)&={\dot{v}_l(t)}+\sum_{j \in \sN_i} \black{\bar{h}_{ij}^{(\tau)}}(\chi_i(t-\tau(t)),\chi_j(t-\tau(t)),t\big)\nonumber \\
&+\bar{\underline{h}}_{il}\big(\chi_i(t),\chi_l(t),t\big)
\end{align}
where the coupling functions $\textcolor{black}{\bar{h}}_{ij}^{(\tau)}:\R^4\times\R^4\times\R_+\rightarrow\textcolor{black}{\R^2}$ and $\textcolor{black}{\bar{\underline{h}}}_{il}:\R^4\times\R^4\times\R_+\rightarrow\textcolor{black}{\R^2}$ can be nonlinear and are smooth and where all the agents are affected by the same delay {as in \cite{7486039}}. In what follows, we let: $h_{ij}^{(\tau)}\left(\chi_i(t-\tau(t)),\chi_j(t-\tau(t)),t\right):= \left[ 0_{1\times 2}, \bar{h}_{ij}^{(\tau)}\left(\chi_i(t-\tau(t)),\chi_j(t-\tau(t)),t\right)^T\right]^T$, and $\underline{h}_{il}\left(\chi_i(t),\chi_l(t),t\right):= \left[ 0_{1\times 2}, \bar{\underline{h}}_{il}\left(\chi_i(t),{\chi_l(t)},t\right)^T\right]^T$. Then, with the following results we establish a sufficient condition for $\mathcal{L}_{\infty}$-sIOS  of the closed-loop dynamics (\ref{hpmodel_linearized:1}) - (\ref{controlnu}). The result are stated in terms of the following matrix 
\begin{equation}\label{eqn:T_matrix}
T_i=\left[\begin{matrix}
I_2 & \alpha_iI_2\\ 0_{2\times 2} & I_2
\end{matrix}\right], \ \ \ \alpha_i >0
\end{equation}
\begin{Proposition}\label{prop: robotic}
Consider the network of mobile agents (\ref{hpmodel_linearized:1}) - (\ref{hpmodel_linearized:2}) controlled by (\ref{controlnu}). Assume that the coupling functions \black{${h_{ij}^{(\tau)}}$, $\underline{h}_{il}$} satisfy, $\forall i=1,\ldots,N$ and $t\ge 0$, the following conditions for some $0<\underline{\sigma}<\bar{\sigma}\textcolor{black}{<+\infty}$ and some vector $[\alpha_1,\ldots,\alpha_N]$ of non-negative constants
\begin{description}
\item[C1] $h_{ij}^{(\tau)}(\chi_i^d(t-\tau(t)),\chi_j^d(t-\tau(t)),t)=\underline{h}_{il}(\chi_i^d(t),\chi_l(t),t)=0$;
\item[C2] $\mu_2\big(\textcolor{black}{T_iA_iT_i^{-1}}+ T_i \partial_1 \underline{h}_{il}(\chi_i,\chi_l,t)T_i^{-1}\big) \le -\bar \sigma, \forall \chi_i,\chi_l\in\R^4$;
\item[C3] $\norm{\sum_{j\in \sN_i} T_i\partial_1 \black{h_{ij}^{(\tau)}}(\chi_i,\chi_j,t)T_i^{-1}}_2+\sum_{j \in \sN_i}\norm{ T_i \partial_2 \black{h_{ij}^{(\tau)}}(\chi_i,\chi_j,t)T_j^{-1}}_2 \le \underline \sigma, \forall \chi_i,\chi_j\in\R^4$.
\end{description}
Then, the network is $\mathcal{L}_{\infty}$-sIOS and in particular, $\forall t \ge 0$:
\begin{align*}
\max_i\abs{\eta_i(t)-\eta_i^d(t)}_2 & \le  K{\max_i}\sup_{-\tau_0 \le s \le 0}\abs{\chi_i(s)-\chi_i^d(s)}_2e^{-\hat \lambda t}\\ &+K\frac{{b_{\max}}}{\bar \sigma-\underline \sigma}{\max_i}\norm{d_i(\cdot)}_{\sL_\infty}, \forall N
\end{align*}
where $K:=\textcolor{black}{\frac{\max_i\{\sigma_{\max}(T_i)\}}{\min_i\{\sigma_{\min}(T_i)\}}}$, $b_{\max}:=\sup_{i,t}\norm{b_i(t)}_2$ and $0<\hat \lambda=\inf_{t\ge 0}\{\lambda|\lambda(t) - \bar \sigma+\underline \sigma e^{\lambda(t)\tau(t)}=0\}$.
\end{Proposition}
\begin{proof} We prove the result via { Proposition \ref{prop:IOS}} and again we omit the explicit dependence of the state variables on time as this is clear from the context. Clearly, the output function is Lipschitz and therefore in order to apply the result we only need to show that the conditions of {Proposition \ref{prop: ISFS}} are satisfied. First, we let $\chi^d(t) :=[{\chi_1^d}^T,\ldots,{\chi_N^d}^T]^T$ be the desired solution of the network, corresponding to the desired formation. We then note that, by means of {C1}, $\chi_i^d(t)$ is a solution of (\ref{eqn:formation_desired_sol}). Also, the fulfilment of {C1}  implies the fulfilment of condition {(i)} in {Proposition \ref{prop: ISFS}}. In order to continue with the proof, for the dynamics (\ref{hpmodel_linearized:1}) - (\ref{hpmodel_linearized:2}) we consider, $\forall i=1,\ldots,N$, the coordinate transformation $\widetilde{\chi}_i(t):=T_i\chi_i(t)$ with $T_i$ defined in (\ref{eqn:T_matrix}). {In particular, by writing (\ref{hpmodel_linearized:1}) - (\ref{hpmodel_linearized:2}) in these new coordinates, one can note that  {C2} and {C3} are equivalent to (ii) and (iii) of Proposition \ref{prop: ISFS}.} Therefore we have:
\begin{align*}
\max_i\abs{\widetilde{\chi}_i(t)-\widetilde{\chi}_i^d(t)}_2  &\le {\max_i}\sup_{-\tau_0 \le s \le 0}\abs{\widetilde{\chi}_i(s)-\widetilde{\chi}_i^d(s)}_2e^{-\hat \lambda t}\\
& +\max_i\{\sigma_{\max}(T_i)\}\frac{{b_{\max}}}{\bar \sigma-\underline \sigma}{\max_i}\norm{d_i(\cdot)}_{\sL_\infty}
\end{align*}
$\forall N$, with $\widetilde{\chi}^d_{i}(t)$ being the desired solution of robot $i$ in the new  coordinates. To obtain the upper bound we  used the fact that:  (i) $\norm{TB(t)}_G \le \norm{T}_G\norm{B(t)}_G \le \max_i\{\sigma_{\max}(T_i)\}b_{\max}$ ({$T$ and $B(t)$ are the $Nn\times Nn$ block diagonal matrices having on their main diagonal the blocks $T_i$'s and $b_i(t)$'s, respectively}); (ii)
${\max_{i}} \vert \widetilde{\chi}_{i}(t)-\widetilde{\chi}_{i}^d(t) \vert_2 \ge  \underline{\lambda} {\max_{i}} \vert \chi_{i}(t)-\chi_{i}^d(t) \vert_2$,  $\underline{\lambda} := \min_i\{\sigma_{\min}(T_i)\}$; (iii) ${\max_{i}}\sup_{-\tau_0 \le s \le 0}\abs{\widetilde{\chi}_i(s)-\widetilde{\chi}_i^d(s)}_2 \le  \bar{\lambda} {\max_{i}}\sup_{-\tau_0 \le s \le 0}\abs{\chi_i(s)-\chi_i^d(s)}_2$, $\bar{\lambda} := \max_{i}\{\sigma_{\max}(T_i)\}$. 
{Therefore, we get 
\begin{align*}
{\max_{i}}\abs{\chi_i(t)-\chi_i^d(t)}_2 &  \le{\frac{\bar{\lambda}}{\underline{\lambda}}}{\max_{i}}\sup_{-\tau_0 \le s \le 0}\abs{\chi_i(s)-\chi_i^d(s)}_2e^{-\hat \lambda t}\\&+{\frac{\bar{\lambda}}{\underline{\lambda}}}\frac{{b_{\max}}}{\bar \sigma-\underline \sigma}{\max_{i}}\norm{d_i(\cdot)}_{\sL_\infty}, \forall N
\end{align*}
proving the result as \black{$\abs{\eta_i(t)-\eta_i^d(t)}_2 \le \abs{\chi_i(t)-\chi_i^d(t)}_2$}.}
\end{proof}

\subsection{Scalability in Cohen-Grossberg recurrent neural networks}\label{Sec:neural_networks}
We now consider the problem of designing scalable neural networks. In particular, we focus on Cohen-Grossberg neural network, which are widely used for e.g. {pattern recognition, associative memories \cite{4666771}} and have Hopfield neural networks as a special case. The model we consider is
\begin{align}\label{cohen}
\dot{x}_i&=p_i(x_i(t))(-c_i(x_i(t))+\sum_{j=1}^{N}a_{ij}g_j(x_j(t))\nonumber \\
&+\sum_{j=1}^Nb_{ij}{g_j^{(\tau)}}(x_j(t-\tau(t)))+{u_i}+d_i(t))
\end{align}
$i=1,\ldots,N$, where $x_i(t) \in \R$ is the state of the $i$-th neuron, $\tau(t) \le \tau_0$ is the time varying transmission delay associated to information transmission, $p_i(\cdot)$ is the amplification function which is assumed to be positive and bounded with $\underline{p}\le p_i(x_i)\le\bar{p}, \forall i, x_i$, {$g_i(\cdot)$ ($g_i^{(\tau)}(\cdot)$) is the activation function of the $i$-th neuron for delay free (delayed) connection}, \textcolor{black}{$a_{ij}\in\R$,  $b_{ij}\in\R$} are the corresponding neuron connection weights, $u_i$ is the (possible) constant input to the neuron and $d_i(t)$ is the exogenous disturbance. The disturbance can model environmental perturbations, adversarial attacks and data biases. In what follows, {$x^*:=[x_1^\ast,\ldots,x_N^\ast]^T$} is the desired network equilibrium point for a given input $u=[u_1,\ldots,u_N]^T$.
\begin{Proposition}\label{prop: cohen}
Consider the recurrent neural network (\ref{cohen}). Assume that there exist some $\bar{\sigma},\underline{\sigma}>0$ such that, $\forall i=1,\ldots,N$ and $\forall t\ge 0$:
\begin{description}
\item[C1] ${-c_i(x_i^*)+\sum_{j=1}^{N}a_{ij}g_j(x_j^*)+\sum_{j=1}^Nb_{ij}{g_j^{(\tau)}}(x_j^*)+u_i=0}$;
\item[C2] $-\partial_1c_i(x_i)+a_{ii}\partial_1g_i(x_i)+\sum_{j\ne i}\abs{a_{ij}\partial_1 g_j(x_j)}\le -\bar{\sigma}$, $\forall x_i, x_j\in\R$;
\item[C3] $\sum_{j=1}^N \abs{b_{ij}\partial_1{g_j^{(\tau)}(x_j)}} \le \underline{\sigma}$, $\forall x_j\in\R$;
\item[C4] $\bar{p}\underline{\sigma}<\underline{p}\bar{\sigma}$.
\end{description}
Then, the network is $\sL_{\infty}$-sISS. In particular, $\forall t\ge0$:
\begin{align*}
{\max_{i}}\abs{x_i(t)-x_i^*} &\le {\max_{i}}\sup_{-\tau_0 \le s \le 0}\abs{x_i(s)-x_i^*}e^{-\hat \lambda t}\\
&+{\frac{\bar{p}}{\underline{p}\bar \sigma-\bar{p}\underline \sigma}}{\max_{i}}\norm{d_i(\cdot)}_{\sL_\infty}, \forall N
\end{align*}
where $0<\hat \lambda=\inf_{t\ge 0}\{\lambda|\lambda(t) - \underline{p}\bar \sigma+\bar{p}\underline \sigma e^{\lambda(t)\tau(t)}=0\}$.
\end{Proposition}
\proof We start with noticing that condition {C1}, together with the fact that the $p_i(\cdot)$'s are {bounded}, implies that {$\dot{x}_i^*(t)=0=p_i(x_i^\ast(t))(-c_i(x_i^*(t))+\sum_{j=1}^{N}a_{ij}g_j(x_j^*(t))+\sum_{j=1}^Nb_{ij}{g_j^{(\tau)}}(x_j^*(t-\tau(t)))+u_i) = p_i(x_i(t))(-c_i(x_i^*)+\sum_{j=1}^{N}a_{ij}g_j(x_j^*)+\sum_{j=1}^Nb_{ij}{g_j^{(\tau)}}(x_j^*)+u_i)$}, with $x_i(t)$ being a solution of (\ref{cohen}). Hence: 
\begin{align*}
\dot{x}_i(t)-\dot{x}^*_i(t) & = p_i(x_i(t))\Big(-c_i(x_i(t)) + c_i(x_i^*) \\
&+ \sum_{j=1}^{N}a_{ij}g_j(x_j(t))-\sum_{j=1}^{N}a_{ij}g_j(x_j^*)\\
& +\sum_{j=1}^Nb_{ij}{g_j^{(\tau)}}(x_j(t-\tau(t)))-\sum_{j=1}^Nb_{ij}{g_j^{(\tau)}}(x_j^*)+d_i(t)\Big)
\end{align*} 
Again, we let $z(t)=[z_i^T(t),\ldots,z_N^T(t)]^T$, $z_i(t)=x_i(t)-x_i^*$. Then, the dynamics for $z(t)$ can be written as 
$$\dot{z}(t)=P(x)\big(A(t)z(t)+H(t)z(t-\tau)+d(t)\big)$$ 
where $P(x)$ is the diagonal matrix having $p_i(x_i(t))$, $i=1,\ldots,N$ on its main diagonal and where $A(t)$ is a $N\times N$ matrix having entries defined, $\forall i,j=1,\ldots,N$, as: 
\begin{description}
\item[•] $A_{ii}(t):=\int_0^1 -\partial_1c_i(\eta x_i+ (1-\eta)x_i^*)d\eta+a_{ii}\int_0^1 \partial_1g_i(\eta x_i +(1-\eta) x_i^*)d\eta$;
\item[•] $A_{ij}(t):=a_{ij}\int_0^1 \partial_1 g_j(\eta x_j+(1-\eta)x_j^*)d\eta$.
\end{description}
Also, $d(t)=[d_1^T(t),\ldots,d_N^T(t)]^T$ and $H(t)=(H_{ij})_{i,j=1}^N$ is the $N\times N$ matrix with the elements: (i) $H_{ii}(t):=b_{ii}\int_0^1 \partial_1{g_i^{(\tau)}}(\eta x_i + (1-\eta)x_i^*)d\eta$; (ii) $H_{ij}(t):=b_{ij}\int_0^1 \partial_1{g_j^{(\tau)}}(\eta x_j + (1-\eta)x_j^*)d\eta$. We then consider the Dini derivative of $\abs{z(t)}_\infty$ and this yields
\begin{equation}\label{dini}
\begin{split}
&D^+\abs{z(t)}_\infty := \limsup_{h\rightarrow 0^+}\frac{1}{h}(\abs{z(t+h)}_\infty-\abs{z(t)}_\infty)\\
=& \limsup_{h\rightarrow 0^+}\frac{1}{h}\big(\vert z(t)+hP(x)A(t)z(t)+hP(x)H(t)z(t-\tau(t))\\ &+hP(x)d(t)\vert_\infty-\abs{z(t)}_\infty\big)\\
\le &\limsup_{h\rightarrow 0^+}\frac{1}{h}\big(\norm{I+hP(x)A(t)}_\infty-1\big)\abs{z(t)}_\infty \\ &+\bar{p}\norm{H(t)}_\infty\sup_{t-\tau(t) \le s \le t}\abs{z(s)}_\infty+\bar{p}{\sup_t\abs{d(t)}_\infty}\\
=& \mu_\infty\big(P(x)A(t)\big)\abs{z(t)}_\infty+\bar{p}\norm{H(t)}_\infty \sup_{t-\tau(t) \le s \le t}\abs{z(s)}_\infty\\ &+\bar{p}{\sup_t\abs{d(t)}_\infty}
\end{split}
\end{equation} 
where we used the fact that $\norm{P(x)}_\infty=\max_i\abs{p_i(x_i)}\le\bar{p}$. Moreover, by definition of $\mu_{\infty}(\cdot)$, we get: 
\begin{align*}
&\mu_\infty\big(P(x)A(t)\big) :=  \max_i\Big\{p_i(x_i)\sum_{j\ne i}\abs{A_{ij}(t)}+p_i(x_i)A_{ii}(t) \Big\}\\
&\le  \max_i\bigg\{\int_0^1 p_i(x_i) \Big(\sum_{j\ne i}\abs{a_{ij} \partial_1 g_j(\eta x_j+(1-\eta)x_j^*)}\\
& -\partial_1c_i(\eta x_i+ (1-\eta)x_i^*)+a_{ii} \partial_1g_i(\eta x_i +(1-\eta) x_i^*)\Big) d\eta\bigg\}
\end{align*}
Then, from the above expression, {C2} implies that {$\mu_\infty (P(x)A(t))\le \max_i\{-p_i(x_i)\bar{\sigma}\}\le -\underline{p}\bar{\sigma}$}. Instead, from { C3} we have that $\norm{H(t)}_{\infty} \le \underline{\sigma}$. Hence (\ref{dini}) becomes: 
\begin{equation*}
\begin{split}
&D^+\abs{z(t)}_\infty \le \\
&-\underline{p}\bar{\sigma}\abs{z(t)}_\infty+\bar{p}\underline{\sigma}\sup_{t-\tau(t) \le s \le t}\abs{z(s)}_\infty+\bar{p}{\sup_t\abs{d(t)}_\infty}
\end{split}
\end{equation*}
Finally, { C4} makes it possible to apply {Lemma \ref{prop:halanay}} and this implies the desired upper bound.
\endproof

\subsubsection*{Hopfield neural networks}
Hopfield neural networks  are a special case of (\ref{cohen}) when $p_i(x_i)=1$ {and $c_i(x_i(t))=c_ix_i(t)$}, $\forall i$. The resulting model is then: 
\begin{align}\label{neural}
\dot{x}_i&=-c_i x_i(t)+\sum_{j=1}^{}a_{ij}g_j(x_j(t))\nonumber \\ &+\sum_{j=1}^Nb_{ij}{g_j^{(\tau)}}(x_j(t-\tau(t)))+{u_i}+d_i(t)
\end{align}
$i=1,\ldots,N$. We let again $x^\ast:=[x_1^\ast,\ldots,x_N^\ast]^T$ be the desired  equilibrium and give the following:
\begin{Corollary}\label{prop: neural}
Consider the Hopfield recurrent neural network (\ref{neural}). Assume that there exist some $\bar{\sigma},\underline{\sigma}>0$ such that, $\forall i=1,\ldots,N$ and $\forall t\ge 0$:
\begin{description}
\item[C1] ${-c_ix_i^*+\sum_{j=1}^{N}a_{ij}g_j(x_j^*)+\sum_{j=1}^Nb_{ij}{g_j^{(\tau)}}(x_j^*)+u_i=0}$;
\item[C2] $-c_i+a_{ii}\partial_1g_i(x_i)+\sum_{j\ne i}\abs{a_{ij}\partial_1g_j(x_j)} \le -\bar{\sigma}, \forall x_i, x_j\in\R$;
\item[C3] $\sum_{j=1}^N\abs{b_{ij}\partial_1{g_j^{(\tau)}}(x_j)} \le \underline{\sigma}$, $\forall x_j\in\R$;
\item[C4] {$\bar{\sigma} > \underline{\sigma}$}.
\end{description}
Then, the network is $\sL_{\infty}$-sISS. In particular, $\forall t\ge0$: 
\begin{align*}
\max_i\abs{x_i(t)-x_i^*} &\le {\max_i}\sup_{-\tau_0 \le s \le 0}\abs{x_i(s)-x_i^*}e^{-\hat \lambda t}\\
&+\frac{1}{\bar \sigma-\underline \sigma}{\max_i}\norm{d_i(\cdot)}_{\sL_\infty}, \forall N
\end{align*}
where $0<\hat \lambda=\inf_{t\ge 0}\{\lambda|\lambda(t) - \bar \sigma+\underline \sigma e^{\lambda(t)\tau(t)}=0\}$.
\end{Corollary}
\begin{proof}
Follows directly from Proposition \ref{prop: cohen}.
\end{proof}
\begin{Remark}
Essentially, Corollary \ref{prop: neural} states that, if there are disturbances on the network, and the matrix measure induced by $\abs{\cdot}_{\infty}$ is considered to study stability, {then not only the system is stable, as shown in \cite{CAO2014165}, but it is also scalable}.
\end{Remark}

\section{Numerical examples}
 {The code, data and parameters to replicate the results are at: \url{https://github.com/GIOVRUSSO/Control-Group-Code}}.
\subsubsection{Designing scalable formations}\label{sec:robotics_example}
We start with illustrating the use of Proposition \ref{prop: robotic} to design protocols guaranteeing $\mathcal{L}_{\infty}$-sIOS of network (\ref{hpmodel_linearized:1}) - (\ref{hpmodel_linearized:2}). As in \cite{lawton}, the parameters of the robots in (\ref{hpmodel}) are $m_i=10.1kg, I_i=0.13kgm^2, l_i=0.12m$, $\forall i$. We recall that, following (\ref{hpmodel_linearized:2}), the output of each robot is $\eta_i = C_i\chi_i$ and hence we consider the case where robots have access to the hand position of their neighbours and to the reference provided by the virtual leader. The coupling functions for the protocol (\ref{controlnu}) considered here are of the form
\begin{equation}\label{robotics_control_protocol}
\begin{split}
&\bar{h}_{ij}^{(\tau)}(\chi_i(t-\tau(t)),\chi_j(t-\tau(t)),t\big)=K_p\left(\left[\begin{matrix}
\chi_{j,1}-\chi_{i,1} \\ \chi_{j,2}-\chi_{i,2}\end{matrix}\right]-\delta_{ji}^d\right) \\
&\bar{\underline{h}}_{il}\big(\chi_i(t),\chi_l(t),t\big)=\\
& K_{pl}\left(\left[\begin{matrix}
\chi_{l,1}-\chi_{i,1} \\ \chi_{l,2}-\chi_{i,2}\end{matrix}\right]-\delta_{li}^d\right) +K_{vl}\left[\begin{matrix}
\chi_{l,3}-\chi_{i,3} \\ \chi_{l,4}-\chi_{i,4}\end{matrix}\right]
\end{split}
\end{equation}
where $\delta_{ji}^d$ denotes the desired offset of agent $i$ from agent $j$, $\delta_{li}^d$ is the desired offset of agent $i$ from virtual leader and $K_p, K_{pl}, K_{vl}$ are $2 \times 2$ diagonal matrices. The offsets are set to achieve the desired formation pattern where the (hand position of the) robots move, at a constant linear speed, in concentric {\em circles} (with the $k$-th circle consisting of $4k$ robots) following the trajectory provided by the virtual leader (see \textcolor{black}{Figure \ref{formation_configuration}}). In the same figure, the desired formation is illustrated when the robots are arranged in a formation  of $3$ concentric circles.

\textcolor{black}{In our first set of simulations, described next, a given robot on the $k$-th circle is connected to: (i) the robots immediately ahead and behind on the same circle; (ii) the closest robot on circle $k+1$ (if any) and the closest robot on {circle} $k-1$ (if any).} 
\begin{figure}
\centering
\includegraphics[width=0.75\columnwidth]{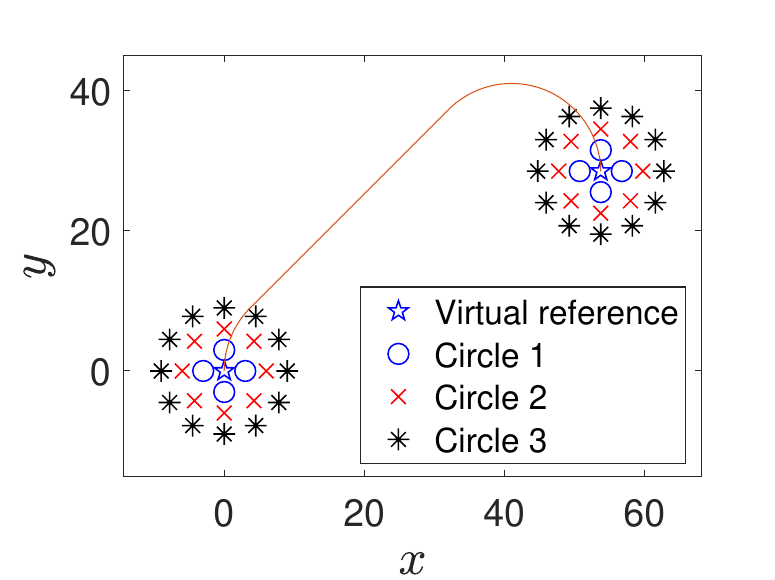}
\caption{Trajectory of the hand position by virtual leader and an example of formation consisting of $3$ concentric circles.}
\label{formation_configuration}
\end{figure}
Also, in the simulations \textcolor{black}{one robot} in the inner circle (i.e. circle $1$) is affected by the disturbances \textcolor{black}{$[d^f,d^{q}]^T=[2\sin(t)e^{-0.2t},{2}\sin(t)e^{-0.2t}]^T$} and the delay is set to $\tau(t)=\tau=0.1s$. Given this set-up, we note that {C1} of {Proposition \ref{prop: robotic}} is guaranteed by the definition of the coupling functions (\ref{robotics_control_protocol}). Therefore, the only conditions that need to be verified are {C2} and {C3}. These can be satisfied by properly choosing $K_p, K_{pl}, K_{vl}$ in (\ref{controlnu}). In particular, we found numerically that the conditions were satisfied by picking the control weights $\black{K_p=0.035I_2}, K_{pl}=0.7I_2, K_{vl}=I_2$, $\forall i,j$. 

With these parameters, \textcolor{black}{we first investigate how the maximum deviation of the hand position (from its desired position) changes as a function of the number of agents in the formation. To this aim, starting with a formation of $1$ circle, we repeatedly simulated the formation by increasing, at each simulation, the number of circles. Then, for each simulation, we recorded the maximum deviation experienced on each circle and finally plot the maximum deviation on each circle across all the simulations. The result of this process is illustrated in Figure \ref{max_deviation} (left panel). Such a figure clearly shows that the disturbance on the first circle is not amplified across the other circles, in accordance with our theoretical predictions.}
\begin{figure}
\centering
\includegraphics[width=0.49\columnwidth]{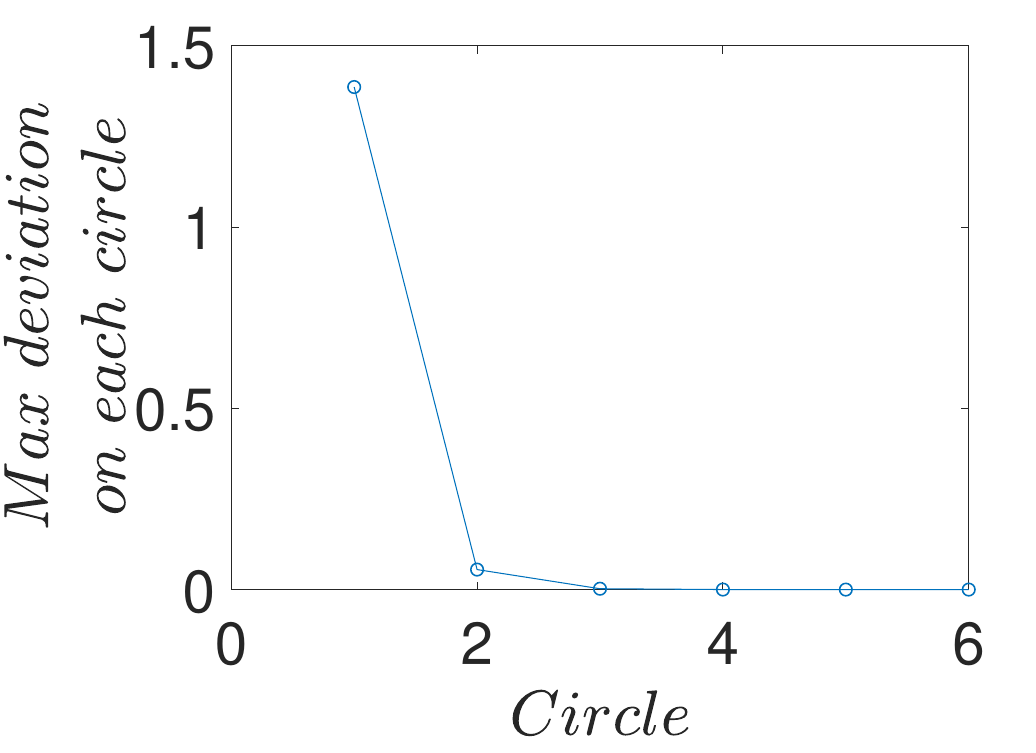}
\includegraphics[width=0.49\columnwidth]{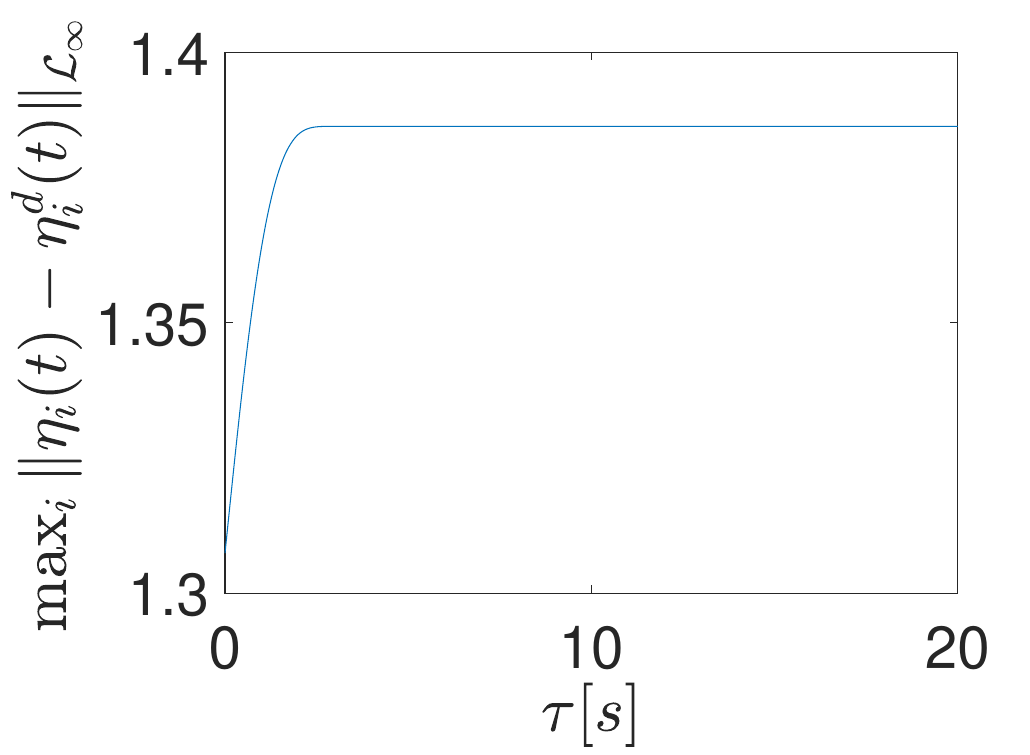}
\caption{\textcolor{black}{Left panel - maximum deviation of the hand position from the desired position (in meters). The figure was obtained by: (i) simulating the network for different number of circles; (ii) recording, at each simulation, the largest deviation on a given circle; (iii) plotting the largest deviations on each circle across all the simulations.} Right panel - maximum deviation of the hand position (in meters) across all the robots as a function of $\tau$ (robots arranged in $6$ circles). In both panels, the same robot was perturbed across all the simulations.}
\label{max_deviation}
\end{figure}
As next step, we also investigate how, for the formation, the maximum deviation of the hand position changes as a function of the delay $\tau$. The results are illustrated in \textcolor{black}{Figure \ref{max_deviation} (right panel)}, which shows that deviations stay bounded when the delay increases (the figure is obtained for the formation with $6$ circles). 

Finally, with our second set of simulations, we further investigate our scalability conditions by considering again the formation in {Figure \ref{formation_configuration}} this time with $14$ {circles} and: (i) each robot on the $i$-th circle connected to the one ahead and behind on the same circle and  with the closest robot on circle $i-1$ (if any); (ii) one robot on the inner circle (i.e. circle $1$) is affected by the same disturbance used in Figure \ref{max_deviation}; \textcolor{black}{(iii) the delay  set to $\tau = 0.1$}s. We simulated the network first with the same control weights used in Figure \ref{max_deviation} \textcolor{black}{(i.e. satisfying the conditions of Proposition \ref{prop: robotic})} and then with a set of control weights that make the network stable but not scalable. As clearly shown in Figure \ref{nonscalable} \textcolor{black}{(top panel)}, scalability prohibits the amplification of perturbations propagating through the network. Instead, when the network is designed to be \textcolor{black}{stable but not scalable}, the disturbances grow when propagating before being attenuated \textcolor{black}{(see middle panel of Figure \ref{nonscalable})}. \textcolor{black}{Finally, the unstable behavior illustrated in the bottom panel of Figure \ref{nonscalable} has been obtained by considering the same protocol and control weights used in the top panel of the figure but this time with each agent  connected to all the others. }  

\textcolor{black}{Interestingly, the simulations in the bottom panel of Figure \ref{nonscalable} show that, perhaps counter-intuitively, a scalable network can be made unstable if new connections are added. We also note that the unstable behavior numerically observed in Figure \ref{nonscalable} (bottom panel) cannot be explained with our sufficient conditions for scalability. A related phenomenon has been also recently observed in the context of synchronization of diffusively coupled delay-free networks. In particular, in \cite{https://doi.org/10.1002/rnc.3863,PhysRevE.95.042312} it has been  shown how a stable synchronization manifold can be made unstable by increasing the connections between certain nodes. With respect to this, an approach to theoretically explain the behavior observed in the bottom panel of Figure \ref{nonscalable}, might be that of extending the sufficient conditions for desynchronization of  \cite{https://doi.org/10.1002/rnc.3863,PhysRevE.95.042312}\footnote{\textcolor{black}{We leave the study of this open problem for our future research}}. }

\begin{figure}[ht]
\centering
\includegraphics[width=\columnwidth]{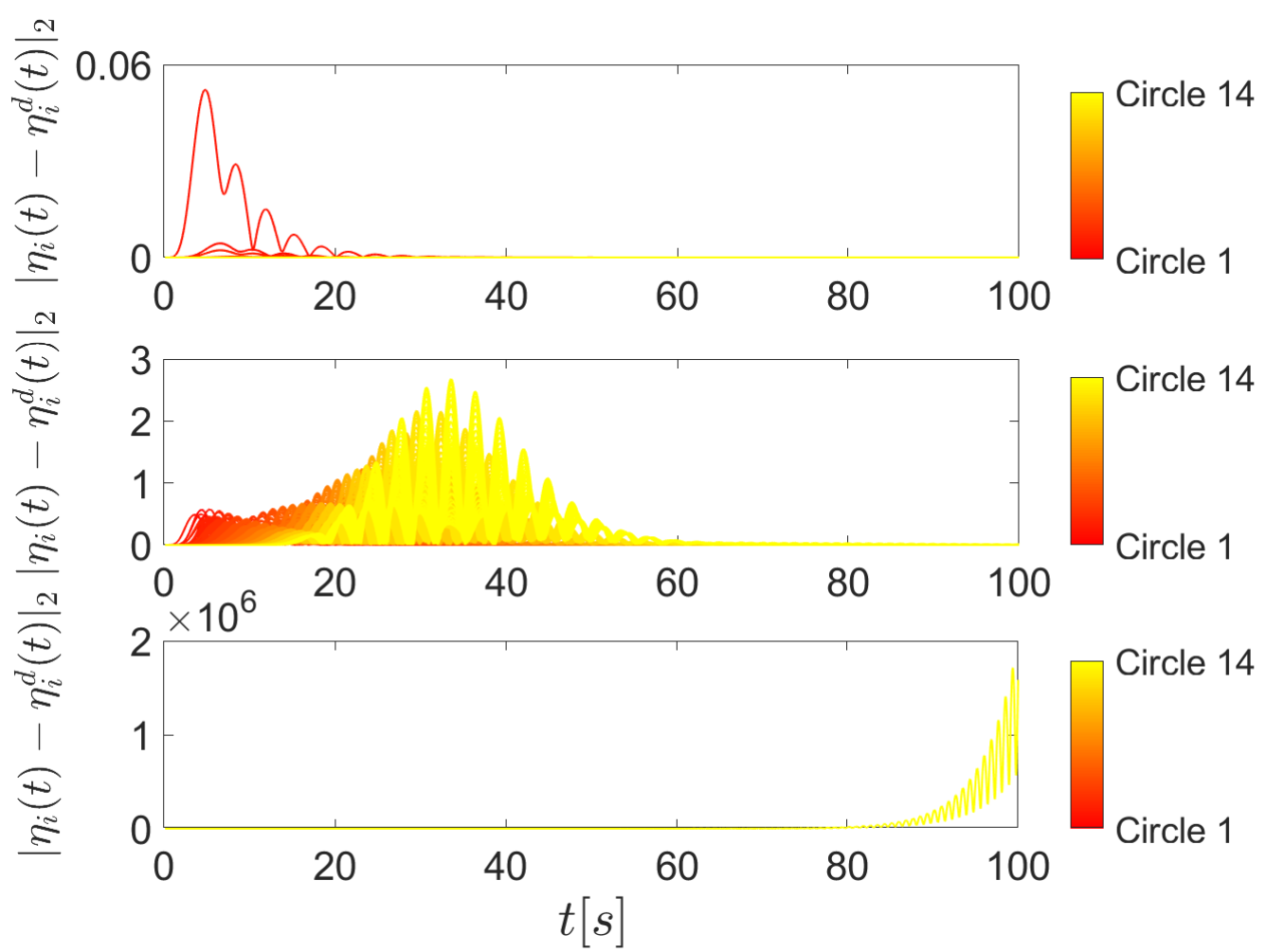}
\caption{\textcolor{black}{\textcolor{black}{Distances of the hand position of each robot from the desired position (in meters).} Top panel - the scalable behaviour for the formation of Figure \ref{formation_configuration} with $14$ {circles}. Middle panel - deviations for the  network when the formation is stable but not scalable. Bottom panel - the scalable network in the top panel becomes unstable by adding links.  Agents on the same circle have the same colour (as in \cite{STUDLI2017157} all the agents except the perturbed one are shown). }}
\label{nonscalable}
\end{figure}

\subsubsection{Designing scalable Hopfield neural networks}
We now turn our attention to the problem of designing a $\mathcal{L}_{\infty}$-sISS scalable Hopfield neural network. We consider (\ref{neural}) with $60$ neurons and: (i) each neuron connected to all the others; (ii) $c_i=10$, $\forall i$; (iii) non-negative weights. All the activation functions are affected by the delay $\tau(t)=\tau=1s$ (i.e. $a_{ij}=0$, $\forall i,j$, in the model). Also, {all the neurons have ${g^{(\tau)}(x)}=\tanh(x)$ as the activation function}. In order to numerically validate the conditions of Corollary \ref{prop: neural}, we first  computed a set of weights verifying conditions C2 - C4 of Corollary \ref{prop: neural}\footnote{weights available at \url{https://github.com/GIOVRUSSO/Control-Group-Code}}. Then, we simulated the network without any disturbance and with {non-negative inputs $u_i$ (the specific inputs are at the  repository)}. This was done to find the unique stable equilibrium towards which the network converges (this is indeed the desired equilibrium, note how C1 is intrinsically satisfied by such equilibrium point). The behaviour of the network when there are no disturbances is illustrated in Figure \ref{equilibrium} (top panel). The network behaviour when the network is affected by disturbances is instead shown in Figure \ref{equilibrium} (bottom panel). In the figure, the deviations of the $x_i(t)$'s with respect to the equilibrium of the unperturbed network are shown. In the figure, $55$  neurons were perturbed at time $5s$ and $15s$, with constant disturbances having a random amplitude between $0$ and $10$ and duration of $1$s. Again, the figure illustrates that disturbances are attenuated within the network, in accordance with the findings. \textcolor{black}{Finally, we also considered the Hopfield network of  Figure \ref{Hopfield_new}, top panel. The activation functions were hyperbolic tangents and these were all affected by delays. We set $\tau(t) ={\tau}=0.1{s}$, $a_{ij}=0, b_{ij}= 15$, for the connections shown in the top panel of the figure and first picked $c_i = 27$, $\forall i$ so that the network was stable \textcolor{black}{but not scalable}. Then, we picked $c_i =32$, which allowed to fulfill the conditions of Corollary \ref{prop: neural} hence making the network scalable. As clearly illustrated in the bottom panels of Figure \ref{Hopfield_new}, in accordance with our  results, when the network is affected by disturbances scalability prohibits their amplification within the network.}

\begin{figure}
  \includegraphics[width=\columnwidth]{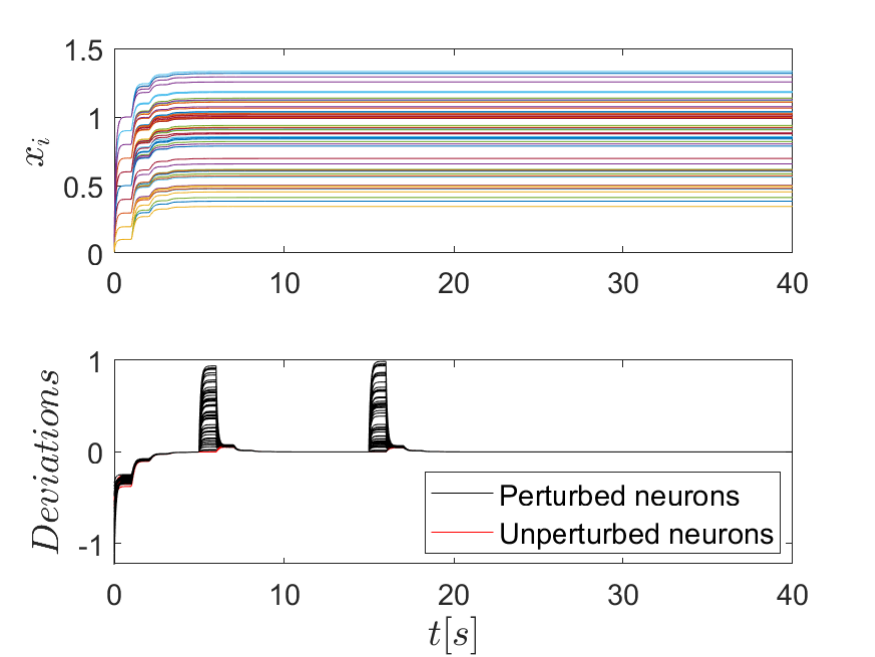}
\caption{Top panel - the network achieves its desired equilibrium when there are no disturbances. Bottom panel - deviations from the equilibrium. In black the deviations for the neurons directly affected by the disturbances (the others in red). Colours online.}
\label{equilibrium}
\end{figure}

\begin{figure}
\begin{minipage}{\columnwidth}
  {\includegraphics[width=\columnwidth, trim={0 180 0 90},clip]{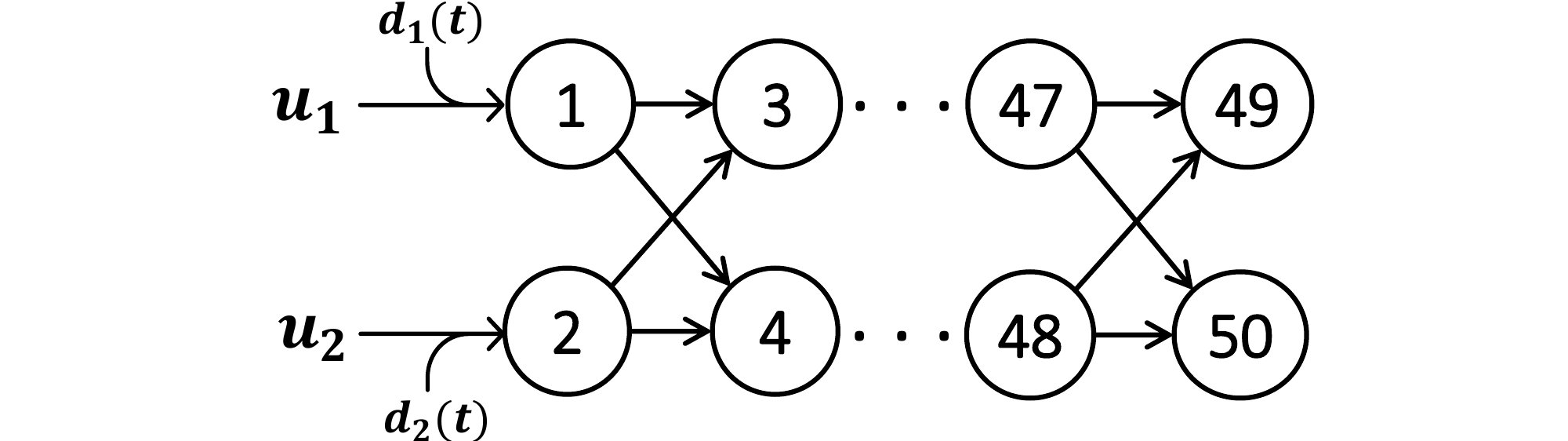}}
\end{minipage}
\begin{minipage}{\columnwidth}
  {\includegraphics[width=1\columnwidth]{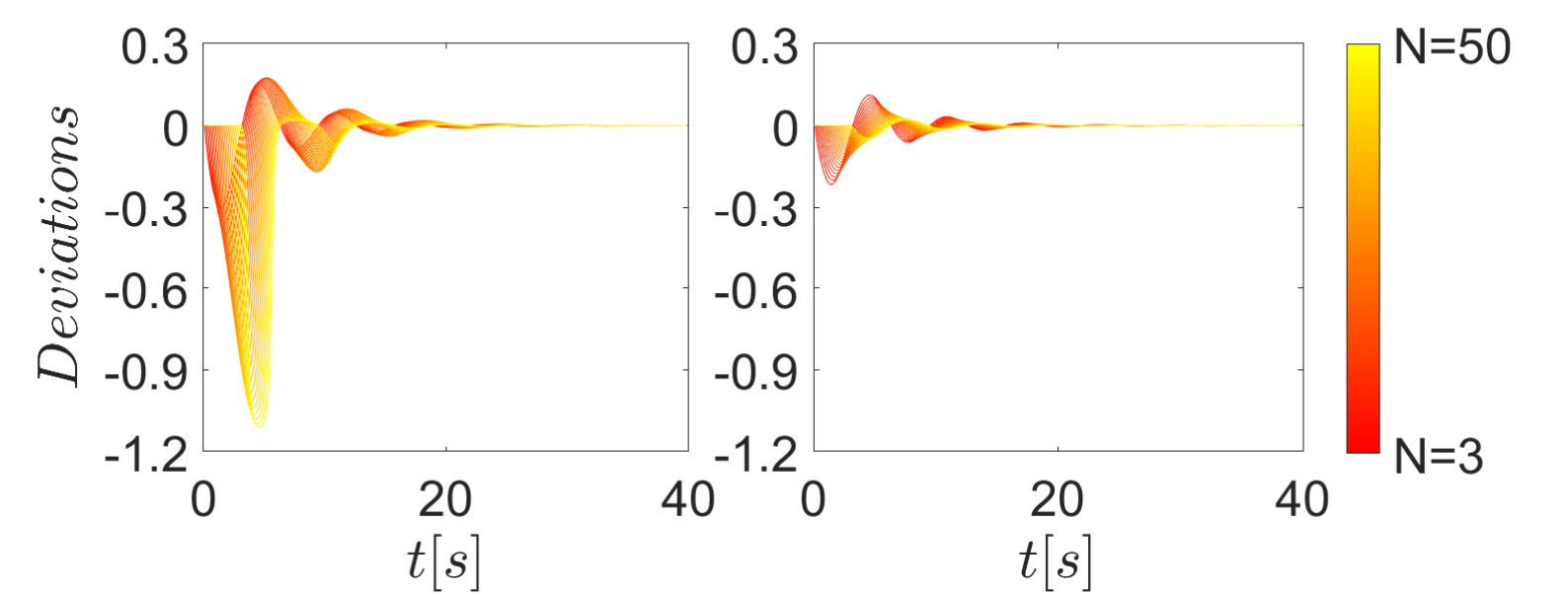}}
\end{minipage}
\caption{\textcolor{black}{Top panel - Hopfield network. Bottom panels - deviations from the equilibrium when the network is stable \textcolor{black}{but not scalable} (left panel) and when this is designed to be scalable (right panel). In all the simulations, neuron $1$ and $2$ were affected by the disturbance {$-10\sin(t)e^{-0.2t}$}. Also, $u_1 = 7$ and {$u_2 = 1$}.}}
\label{Hopfield_new}
\end{figure}

\section{Conclusions and discussion}\label{sec:conclusions}
We considered networks of possibly nonlinear heterogeneous agents coupled via possibly nonlinear protocols affected by delays and disturbances. For these networks, after introducing the notions of $\mathcal{L}_{\infty}$-sISS and $\mathcal{L}_{\infty}$-sIOS, we presented two sufficient conditions to assess these properties. The conditions can be turned into design guidelines  and we used our results to: (a) design distributed control protocols able to guarantee both tracking of a time-varying reference and $\mathcal{L}_{\infty}$-sIOS; (b)  design the activation functions (and their weights) of certain recurrent neural networks so that these are $\mathcal{L}_{\infty}$-sISS. The effectiveness of the results was illustrated via  simulations. Besides considering heterogeneous delays, future work might involve: (i) \textcolor{black}{devising scalability conditions that take into account bounds on the delays and investigate the conservativeness of our conditions}; (ii) \textcolor{black}{investigating, inspired by certain recent literature (see e.g. \cite{pmlr-v120-bonassi20a,pmlr-v120-revay20a}) on stable RNNs to model nonlinear input/ouput sequences, scalability for complex, time-varying, input/output patterns}; (iii) stochastic disturbances. \textcolor{black}{Finally, motivated by the numerical findings reported in the bottom panel of Figure \ref{nonscalable}, we are currently working towards devising sufficient conditions for the loss of scalability.}
\textcolor{black}{\subsubsection*{Acknowledgments}
The authors are grateful to the anonymous reviewers and the AE for their constructive feedback.}
\appendix
The derivations are inspired from \cite{lawton}, where the same model was considered but without disturbances. We let $x_i=[p^x_i,p^y_i,v_i,\theta_i,\omega_i]^T$, \textcolor{black}{$\bar u_i=[F_i,Q_i]^T$, $\bar d_i=[d^f_i,d^q_i]^T$} and aim at controlling the hand position $\eta_i=p_i+l_i\left[\cos{\theta_i}, \sin{\theta_i}\right]^T$, where $l_i$ is the distance of the hand position from the inertial position. We differentiate $\eta_i$ twice to get
\begin{align*}
\begin{split}
\ddot{\eta}_i=&\left[\begin{matrix} -v_i\omega_i\sin{\theta_i}-l_i\omega_i^2\cos{\theta_i} \\ v_i\omega_i\cos{\theta_i}-l_i\omega_i^2\sin{\theta_i} \end{matrix}\right] + \left[\begin{matrix} \frac{1}{m_i}\cos{\theta_i} & -\frac{l_i}{I_i}\sin{\theta_i} \\ \frac{1}{m_i}\sin{\theta_i} & \frac{l_i}{I_i}\cos{\theta_i}\end{matrix}\right]\left[\begin{matrix} F_i \\ Q_i\end{matrix}\right]\\
+&\left[\begin{matrix} \frac{1}{m_i}\cos{\theta_i} & -\frac{l_i}{I_i}\sin{\theta_i} \\ \frac{1}{m_i}\sin{\theta_i} & \frac{l_i}{I_i}\cos{\theta_i}\end{matrix}\right]\left[\begin{matrix} d^f_i \\ d^q_i\end{matrix}\right]
\end{split}
\end{align*}
Considering the diffeomorphism $\chi_i=T_i(x_i)=[p^x_i+l_i\cos{\theta_i}, p^y_i+l_i\sin{\theta_i}, v_i\cos{\theta_i}-l_i\omega_i\sin{\theta_i}, v_i\sin{\theta_i}+l\omega_i\cos{\theta_i}, \theta_i]^T = [\chi_{i,1}, \chi_{i,2}, \chi_{i,3}, \chi_{i,4}, \chi_{i,5}]^T$, we get $\dot{\chi}_{i,1}= \chi_{i,3}$, $\dot{\chi}_{i,2}= \chi_{i,4}$ and
\begin{align}\label{hpmodel}
    \begin{split}
        \left[\begin{matrix} \dot{\chi}_{i,3}\\ \dot{\chi}_{i,4}\end{matrix}\right]=& \left[\begin{matrix} -v_i\omega_i\sin{\chi_{i,5}}-l_i\omega_i^2\cos{\chi_{i,5}} \\ v_i\omega_i\cos{\chi_{i,5}}-l_i\omega_i^2\sin{\chi_{i,5}} \end{matrix}\right]\\
        +&\left[\begin{matrix} \frac{1}{m_i}\cos{\chi_{i,5}} & -\frac{l_i}{I_i}\sin{\chi_{i,5}} \\ \frac{1}{m_i}\sin{\chi_{i,5}} & \frac{l_i}{I_i}\cos{\chi_{i,5}}\end{matrix}\right] \textcolor{black}{\bar u_i}\\
        +&\left[\begin{matrix} \frac{1}{m_i}\cos{\chi_{i,5}} & -\frac{l_i}{I_i}\sin{\chi_{i,5}} \\ \frac{1}{m_i}\sin{\chi_{i,5}} & \frac{l_i}{I_i}\cos{\chi_{i,5}}\end{matrix}\right] \textcolor{black}{\bar d_i}\\
        \dot{\chi}_{i,5}=&-\frac{1}{2l_i} \chi_{i,3}\sin{\chi}_{i,5}+\frac{1}{2l_i} \chi_{i,4}\cos{\chi}_{i,5}\\  
          \end{split}
\end{align}
with $\eta_i= [\chi_{i,1}, \chi_{i,2}]^T$. The feedback linearizing control is
\begin{equation}\label{feedback_lin_contr}
\begin{split}
    &  \textcolor{black}{\bar u_i} = \left[\begin{matrix} \frac{1}{m_i}\cos{\chi_{i,5}} & -\frac{l}{I_i}\sin{\chi_{i,5}} \\ \frac{1}{m_i}\sin{\chi_{i,5}} & \frac{l}{I_i}\cos{\chi_{i,5}}\end{matrix}\right]^{-1}\cdot\\
    &\cdot\left( \textcolor{black}{\bar \nu_i}(t)-\left[\begin{matrix} -v_i\omega_i\sin{\chi_{i,5}}-l\omega_i^2\cos{\chi_{i,5}} \\ v_i\omega_i\cos{\chi_{i,5}}-l\omega_i^2\sin{\chi_{i,5}} \end{matrix}\right] \right)
\end{split}
\end{equation}
with $\bar\nu_i:=\left[\bar\nu_{i,1}(t) \ \ \bar\nu_{i,2}(t)\right]^T$. Hence,  noticing that the zero dynamics is stable, yields the following  reduced dynamics
\begin{align*}
    \begin{split}
        \dot{\chi}_{i,1}&= \chi_{i,3}\\
        \dot{\chi}_{i,2}&= \chi_{i,4}\\
        \left[\begin{matrix} \dot{\chi}_{i,3}\\ \dot{\chi}_{i,4}\end{matrix}\right]&= \textcolor{black}{\bar \nu_i(t) + \bar b_i(t)\bar d_i(t)} \\
    \end{split}
\end{align*}
with $\eta_i= [\chi_{i,1}, \chi_{i,2}]^T$ and where $\bar b_i(t):= \left[\begin{matrix} \frac{1}{m_i}\cos{\theta_i(t)} & -\frac{l_i}{I_i}\sin{\theta_i(t)} \\ \frac{1}{m_i}\sin{\theta_i(t)} & \frac{l_i}{I_i}\cos{\theta_i(t)}\end{matrix}\right]$. 
This is the dynamics considered in Section \ref{sec:robotics} where $\chi_i:=[\chi_{i,1},\chi_{i,2},\chi_{i,3},\chi_{i,4}]^T$.

\bibliographystyle{IEEEtran}

\end{document}